\documentclass{article}

\usepackage[preprint]{neurips_2022}

\usepackage[utf8]{inputenc}
\usepackage{fullpage}
\usepackage{amsmath,amsfonts,amsthm,amssymb}
\usepackage{url}
\usepackage{color}
\usepackage[usenames,dvipsnames,svgnames,table]{xcolor}
\usepackage[colorlinks=true, linkcolor=red, urlcolor=blue, citecolor=blue]{hyperref}

\usepackage{booktabs}
\usepackage{bm}
\usepackage{stackrel}
\usepackage{enumitem}

\usepackage[parfill]{parskip}

\newcommand{\Ac}{{\mathcal{A}}}

\newcommand{\Ic}{{\mathcal{I}}}

\newcommand{\Pc}{{\mathcal{P}}}

\newcommand{\Tc}{{\mathcal{T}}}

\newcommand{\Xc}{{\mathcal{X}}}

\newcommand{\Nbb}{{\mathbb{N}}}

\newcommand{\Rbb}{{\mathbb{R}}}

\newcommand{\Zbb}{{\mathbb{Z}}}

\newtheorem{definition}{Definition}
\newtheorem{corollary}{Corollary}
\newtheorem{theorem}{Theorem}

\newcommand{\known}[1]{\mbox{$#1^{\text{obs}}$}}
\newcommand{\unknown}[1]{\mbox{$#1^{\text{unobs}}$}}
\newcommand{\re}[1]{\mbox{$#1^{'}$}}

\newcommand{\coeffsubtensorScaled}{Z}
\newcommand{\coeffsubtensorscaled}{z}

\newcommand{\coeffsubtensorAdd}{T}

\newcommand{\TCA}{\mbox{TCA}}
\newcommand{\MCA}{\mbox{MCA}}

\usepackage[utf8]{inputenc} 
\usepackage[T1]{fontenc}    
\usepackage{hyperref}       
\usepackage{url}            
\usepackage{booktabs}       
\usepackage{amsfonts}       
\usepackage{nicefrac}       
\usepackage{microtype}      
\usepackage{xcolor}         
\usepackage{tikz}
\tikzset{>=latex}
\usepackage{pgfplots}
\pgfplotsset{compat=1.16}
\usepackage{caption, subcaption}
\usepackage[ruled,vlined]{algorithm2e}

\SetKwInput{KwInput}{Input}                
\SetKwInput{KwOutput}{Output}              

\title{A Simple and Scalable Tensor Completion Algorithm\\via Latent Invariant Constraint\\for Recommendation System}

\author{
  Tung Nguyen \\
  University of Missouri - Columbia\\
  \texttt{tdn84d@mail.missouri.edu} \\
  \And
  Sang T. Truong \\
  Stanford University \\
  \texttt{sttruong@cs.stanford.edu} \\
  \And
  Jeffrey Uhlmann \\
  University of Missouri - Columbia\\
  \texttt{uhlmannj@missouri.edu} \\
}

\begin{document}

\maketitle

\begin{abstract}
In this paper we provide a latent-variable formulation and solution to the recommender system (RS) problem in terms of a fundamental property that any reasonable solution should be expected to satisfy. Specifically, we examine a novel tensor completion method to efficiently and accurately learn parameters of a model for the unobservable personal preferences that underly user ratings. By regularizing the tensor decomposition with a single latent invariant, we achieve three properties for a reliable recommender system: (1) uniqueness of the tensor completion result with minimal assumptions, (2) unit consistency that is independent of arbitrary preferences of users, and (3) a consensus ordering guarantee that provides consistent ranking between observed and unobserved rating scores. Our algorithm leads to a simple and elegant recommendation framework that has linear computational complexity and with no hyperparameter tuning. We provide empirical results demonstrating that the approach significantly outperforms current state-of-the-art methods.
\end{abstract}

\section{Introduction}
The emphasis on user-specific latent traits has driven the algorithmic development of recommendation systems over the last decade \cite{RS4, RS5, RS3,  RS2, RS6, RS1}. As these latent variables are highly complex in structure for very large datasets, providing a rigorous analysis of the latent vectors based on data relating users, products, and features can be expected to illuminate salient structures. A recommendation system (RS) is intended to infer different user evaluations of products to accurately predict how a given user is likely to evaluate a new product \cite{IRT1}. To achieve this goal, the RS must extrapolate a model that captures the relationships among observed evaluations of products by a given user to those of different users \cite{IRT2}. The key is to ensure that inferences are derived from a common latent space that is fundamental to the evaluation process of different users.  

In this paper we propose unit consistency (UC) as a latent invariant that is common to user evaluations. UC presumes that user evaluations will be invariant up to an arbitrary choice of positive units applied to a set of incommensurate variables defining the state of a system \cite{IRT3}. For example, consider state vector $x$ for a system defined with lengths in meters and rotations in radians per minute, and another state vector $y$ for the same system but with lengths in centimeters and rotations in radians per second. Unit consistency implies that for a nonlinear transformation $\Tc$, $\Tc y$ gives the same result as $\Tc x$ but in the units of $y$. Intuitively, the choice of units should not qualitatively affect the result of a transformation, just the units in which the result is expressed \cite{UC2}. UC fits into the context of RS through the following demonstration. In the context of RS, each user can be presumed to evaluate each product based on an implicit set of personal units applied to each attribute \cite{IRT4}. The RS cannot possibly know the full set of attributes each user evaluates in the rating of each product, let alone infer the units each user applies with respect to each of those attributes \cite{IRT5}. Even a user would unlikely be able to reflect on and identify the precise set of variables that led to their specific rating for a particular product \cite{RS7}. Thus, at a minimum, a RS should be expected to compute ratings in a manner that provides UC with respect to such unknowable units. In other words, the RS should implicitly model user ratings as deriving from an underlying model in which ratings from different users can be explained in terms of each having their own individual set of units that relate to their individual personal preferences \cite{RS7, IR6}.
 
In this paper, we show that regularizing tensor completion algorithm by a small number of latent invariants is sufficient to yield unique and accurate recommendations underpinning a reliable RS. To acquire the solutions satisfying the UC constraints, we define an optimization framework that guarantees uniqueness of normalized tensor and tensor completion results with minimal assumptions. Furthermore, we show that results from our method satisfy consensus ordering property.  Experimental tensor-completion results in both 2D and 3D demonstrate significant performance improvements over prior state-of-the-art recommender system algorithms.

\section{Related Works}
Computationally efficient tensor completion can be formulated as an optimization problem \cite{Tractable1, Tractable2}, where nuclear norm minimization is employed to enforce unitary invariance. To tackle this problem,~\cite{Recht, Boaz, Aaron} utilized singular value decomposition (SVD) method to achieve state-of-the-art probabilistic bounds on the minimum number of entries sufficient for retrieval. Other methods, such as collaborative filtering~\cite{Daniel}, parameter-decrease~\cite{Guangxiang}, Scale Invariant Feature Transform~\cite{Lowe, Gihwi}, and social choice theory~\cite{David}, exploit unitary invariance indirectly via the Moore-Penrose pseudoinverse. However, some applications might require desirable properties beyond unitary invariance, such as physics or robotics \cite{UC1, UC2, UC3}. For example, we might want invariance with respect to the choice of units on key variables \cite{UC1} or with respect to global rotations of the system coordinate \cite{UC2}.

In this paper, problem-specific invariants are identified and enforced to be fundamental, and relaxation of a critical invariant for computational efficiency is only considered as a last resort. Our algorithm takes a simpler iterative structure from the tensor completion procedure and then deterministically generates recommendation scores without consideration of the rank of the tensor. \cite{Anandkumar2014} considered tensor completion that is applicable for a sparse-coding model in the overcomplete regime, where the tensor rank is larger than the dimension. Our work reduces the complexity involved within the parameter space of the tensor rank and allows a fast convergence rate. To estimate latent features in the complex parameter space,~\cite{OwenZhou2000, FangLi2021} also examine a theoretical gradient-based method procedure to estimate the latent parameters. With minimal assumptions, we utilize a convex optimization model based on~\cite{CSA} to retrieve the latent vectors and provide recommendation results directly.

\section{Method}
\subsection{Problem Set Up}
We consider the following RS as a running example. A recommender interacts with the user $i$-th whose feature vector $x_i \in \Xc \subset \Rbb^{d_x}$ is assumed to be determined by an unknown but fixed latent variable $z_i$. For a recommendation framework $A$, we denote the rating of user $u_i$ for product $p_j$ as $r_{u_i, p_j}$. For user $i$, the recommender has access to that user's rating by indexing into a row of a tensor $A$ with ratings as $r_{u_i, p_1}, r_{u_i, p_2}, ..., r_{u_i, p_{P_i}}$ of $P_i$ products $\{p_1, ..., p_{P_i} \} \subseteq \Pc$. Each movie $p_j$ in the set $\Pc$ is characterized by a vector $p_j \in \Pc \subset \Rbb^{d_p}$~\footnote{We overload the index of the movie and its features}. Given a user $i$ with a fixed latent variable $z_i$, we assume that there is a function $f$ parameterized by $z_i$ relating user features and movie features to rating (Figure~\ref{graphical-model}):
\begin{equation}
     r_{u_i,p_j} = f(x_i, p_j; z_i) = A[u_i, p_j, x_i[0], ..., x^i[d_x], p_j[0], ..., p_j[d_p]]
\end{equation}
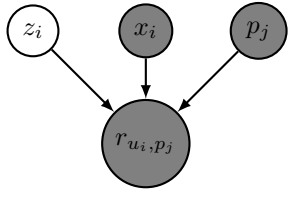
\begin{figure}[t]
    \centering
    \begin{tikzpicture}[
            node distance={15mm}, thick, 
            observed_node/.style={draw, circle, fill=gray},
            unobserved_node/.style={draw, circle},
        ]
        
        \node[unobserved_node]  (1)              {$z_i$}; 
        \node[observed_node]    (2) [right of=1] {$x_i$}; 
        \node[observed_node]    (3) [right of=2] {$p_j$};
        \node[observed_node]    (4) [below of=2] {$r_{u_i,p_j}$} ; 
    
        \draw[->] (1) -- (4);
        \draw[->] (2) -- (4);
        \draw[->] (3) -- (4);
    \end{tikzpicture}
    \caption{Graphical demonstration of data generating process. Shaded and unshaded nodes are observable unobservable variables, respectively.}
    \label{graphical-model}
\end{figure}

The RS aims to recommend a previously-unrated product $j$ to user $i$ that is most likely to give a maximum rating $r_{u_i, p_j}$. Toward this goal, we learn the relationship $f(x_i, p_j; z_i)$ to predict an unrated movie through tensor completion.

Our data has the form of a tensor $A\in \Rbb^{I_1 \times ... \times I_D}$\footnote{ Without loss of generality, we assume for convenience that $A$ is strictly positive.}, and we index $A$ with index vector $i = [i_1, ..., i_D]$:
\begin{equation}
    A[i]  = A[i_1, ..., i_D], \text{ where } i_d \in [D] \cup \emptyset \text{ and } [D] = \{n \in \Nbb: n < D \}
\end{equation}
When $i_d = \emptyset$, we access all elements in dimension $d$ of $A$. Let $\Ic$ be a set of all vectors $i$ with length $D$ that has $k$ non-null components. A $k$-dimensional subtensor $A[i] \in \Rbb^{I_{d_1}\times\dots\times I_{d_k}}$ for $i\in \Ic$ spans all the $k$ null components of $i$ at dimensions $d_1, \dots, d_k$. Let $\Ac_\Ic = \{A[i] | i \in \Ic\}$ be the set of $k$-dimensional subtensors of $A$. If vector $i'$ with $D$-non-null elements satisfies that $A[i']$ is an element of subtensor $A[i]$, we denote $i' \in A[i]$. We define the set of observed and unobserved entries of $A$ as:
\begin{equation}
     \known{A}= \{i \mid A[i] \neq 0\} \text{ and } \unknown{A} = \{i \mid A[i] = 0\} 
\end{equation}
Let $\coeffsubtensorScaled_k \in \Rbb^{|\Ic|}$ be a strictly positive-valued vector of length equal to the number of $k$-dimensional subtensors of $A$. Then $\coeffsubtensorScaled_{k, A[i]}$ is an element of $\coeffsubtensorScaled_k$ associated with subtensor $A[i]$. We provide a more rigorous concept of a latent variable extraction of a tensor $A$. The product $\re{A}=\coeffsubtensorScaled_k *_k A$ is defined as a latent scaling of each $k$-dimensional subtensor $A[i]$ of $A$ to $\re{A}[i]$ of $\re{A}$ as  $\re{A}[i] = \coeffsubtensorScaled_{k, A[i]}\cdot A[i]$.

For example, when $D=3$ and $k=1$:
\begin{equation}
\begin{aligned}
    \Ic & = \bigcup_{[i_1, i_2, i_3]} \{[i_1, i_2, \emptyset], [\emptyset, i_2, i_3], [i_1, \emptyset, i_3]\} \\
    \Ac_\Ic &= \bigcup_{[i_1, i_2, i_3]}\{A[i_1,i_2,\emptyset], A[\emptyset,i_2,i_3], A[i_1,\emptyset,i_3]\}
\end{aligned}
\end{equation}

Our objective is to learn latent variable $z_k$ from $A$ via tensor decomposition. Specifically, we seek the $D$-dimensional tensor $\re{A} \in \Rbb^{I_1\times \cdots \times I_D}_{+}$ and a positive vector $\coeffsubtensorScaled_k$ such that
\begin{equation}
\begin{aligned}
    \re{A} & = A *_k  \coeffsubtensorScaled_k \\
    \prod_{i\in \re{A}\known{ [i']}}\re{A}[i] & = 1, \quad \forall \re{A}[i'] \in \re{A}_{\Ic}
\end{aligned}
\label{constraints}
\end{equation}
meaning that the product of the observed entries of each $k$-dimensional subtensor $\re{A}[i']$ is 1. 

An equivalent objective is to learn the log latent variable. For notational convenience, let $A$ also be the logarithm conversion of $A$, i.e., all known entries are replaced with their logs. Specifically, we seek the $D$-dimensional tensor $\re{A}$ and a latent vector $\coeffsubtensorscaled_k$ such that
\begin{equation}
    \re{A}[i] \equiv A[i] + \sum\limits_{i' \in \Ic:i\in A[i']} \coeffsubtensorscaled_{k,A[i']}, \quad \forall i \in \known{A}  \text{ where } \sum_{i\in \re{A}\known{ [i']}}\re{A}[i] = 0, \quad  \forall \re{A}[i'] \in \re{A}_{\Ic}
\end{equation}
meaning the sum of observed entries of $k$-dimensional subtensor $\re{A}[i']$ equals zero. We conduct the learning process in log latent purely to improve the numerical properties of our algorithm.

\subsection{Algorithms}
We present Algorithm 1 for learning latent variables, with the details of the correctness and efficiency of the algorithm  in~\cite{CSA}, where the authors proved its equivalence to a convex optimization model, with optimal time complexity that is linear in the number of observed entries, $O\left(|\known{A}| \right)$. This time complexity is valid even when we encounter data sparsity as $|\known{A}| \ll |\unknown{A}|$ and is unaffected by the final step of converting back from the log-space solution to the desired solution. Here we present our Learning Latent Invariance algorithm (LLI) as Algorithm~\ref{LLI} and the tensor completion algorithm (TCA) in Algorithm~\ref{TCA}.
\begin{algorithm}[t]
    \DontPrintSemicolon
    \KwInput{$D$-dimensional tensor \textit{A}.}
    \KwOutput{$\re{A}$ and scaling vector $\coeffsubtensorScaled_k$.}
    
    \SetKwFunction{FMain}{LLI}
    \SetKwProg{Fn}{Function}{:}{}
    \Fn{\FMain{$A, k$}}{
    \textbf{- \texttt{Step 1}: Iterative step over constraints}: Initialize $count \leftarrow 0$, variance variable $v\leftarrow 0$, and let $\rho$ be a zero vector of conformant length. For notational convenience, let $A$ also be the logarithm conversion of $A$, i.e., all known entries are replaced with their logs.
    \\
    \For{each subtensor $A[i']\in A_{\Ic}$}{
        \begin{equation}
        \begin{aligned}
            \rho_{i'} &= -\left[ |\known{A[i]}|\right]^{-1} 
            \sum\limits_{i \in \known{A[i']}}\hspace{-4pt} A[i]
            \\
            A[i] & \leftarrow A[i] + \rho_{i'},\, v\leftarrow v+\rho_{i'}^2, \quad \text{for} \quad i \in \known{A[i']}
            \\
            \coeffsubtensorscaled_{k, A[i']}&\leftarrow \coeffsubtensorscaled_{k, A[i']} + \rho_{i'}
        \end{aligned}
        \end{equation}
    }
    \textbf{- \texttt{Step 2}: Convergence}: If $v$ is less than a selected threshold $\epsilon$, then exit loop. Otherwise, set $count \leftarrow count+1$ and return to step 1. 
    
    \KwRet $\re{A} \leftarrow exp \left(A\right)$ and $\coeffsubtensorScaled_k \leftarrow  exp(\coeffsubtensorscaled_k)$.\;}
    \caption{Learning Latent Variables of Tensor $A$}
    \label{LLI}
\end{algorithm}

\begin{algorithm}[t]
    \DontPrintSemicolon
    \KwInput{$D$-dimensional tensor A and $k$.}
    \KwOutput{$\re{A}$}
    \SetKwFunction{FMain}{TCA}
    \SetKwProg{Fn}{Function}{:}{}
    \Fn{\FMain{$A, k$}}{
        \textbf{- \texttt{Step 1}: LLI process.} \\
        $\coeffsubtensorScaled_k \leftarrow LLI(A, k)$ \\
        \textbf{- \texttt{Step 2}: Tensor completion process}: \\
        $\re{A} = A$ \\
        \For{$i \in \unknown{A}$}{
            $\re{A}[i] \leftarrow \prod_{i' \in \Ic: i \in A[i']} \coeffsubtensorScaled_{k,A[i']}^{-1}$ .\;
        }
      \KwRet $\re{A}$.
    }
    \caption{Tensor Completion Algorithm}
    \label{TCA}
\end{algorithm}

\paragraph{A Special Case of Tensor Completion Algorithm} In the case of $D$=$2$ when $A$ is a 2D tensor, i.e., $LLI(A, k$=$1)$, the set of $k$=$1$ subtensors is simply the set of rows and columns. This is a specialized instance of a tensor problem studied in~\cite{RZalgorithm} with all of its dimensions explicitly distinguished for the problem of scaling line products of a tensor to chosen positive values (for which a solution is not guaranteed to exist except in the case we use of all scaling values equal to $1$). For notational convenience, we define the completion function $\MCA(A)$ as a special case of TCA:
\begin{equation}
       \MCA(A) ~\equiv ~ \TCA(A,1) ~~ \mbox{for $D=2$}\, .
\end{equation}
The time and space complexity for $\MCA(A)$ is $O(|\known{A}|)$. Both TCA and MCA directly yield unit-consistent completion algorithms. Here we examine their properties when applied to the RS problem.

\subsection{Theoretical Analyses}
\subsubsection{Uniqueness}
\begin{theorem}\textbf{Uniqueness of $\re{A}$ in Learning Latent Invariance Algorithm:}
    There exists a unique tensor $\re{A}$ for which there exists a strictly positive latent vector $\coeffsubtensorScaled_k$ such that the solution $\re{A} \leftarrow LLP(A, k)$ is unique and and $\known{\re{A}} = \known{A}$. Furthermore, if a positive vector $\coeffsubtensorAdd_k$ satisfies
    \begin{equation}
        \prod\limits_{i' \in \Ic:i \in A[i']}\hspace{-4pt}  \coeffsubtensorAdd_{k, A[i']}  ~= ~1 \qquad \forall i \in \known{A},
    \end{equation} then $\left(\re{A}, \coeffsubtensorScaled_k \circ \coeffsubtensorAdd_k\right)$ is also a solution, where $\circ$ is the Hadamard product.
    \label{uniquess_A^{re}}
\end{theorem}{}
\begin{proof}
    The full proof of the first part is in the paper~\cite{CSA}. For the second part, we have for each $i \in \known{A}$,
    \begin{equation}
         \prod_{i'\in \Ic:i\in A[i']}A[i]\cdot\left(\coeffsubtensorScaled_{k, A[i']}\cdot\coeffsubtensorAdd_{k, A[i']}\right) = \prod_{i' \in \Ic:i\in A[i']}A[i]\cdot\coeffsubtensorScaled_{k, A(i')} 
         =  \re{A}[i] ,
    \end{equation}
    and 
    \begin{equation}
        \prod_{i\in \re{A} \known{[i']}}\re{A}[i] = 1, \quad \forall \re{A}[i'] \in \re{A}_{\Ic}
    \end{equation}
\end{proof}

Although the learned tensor is unique, the latent vectors $z_k$ and $TCA(A,k)$ may not be. We will be able to define a sufficient number and structure of known entries such that the tensor completion result, $TCA(A,k)$, is unique. We refer to a tensor with sufficient known entries to guarantee uniqueness as having {\em full support}. Definition \ref{fully-supported-tensor} formally defines our notion of full support.

\begin{definition} \textbf{Fully Supported Structure:}
    Given $A$, we define that tensor as {\em fully supported} if for every entry $i \in \unknown{A}$, there exist $2^{D}-1$ vectors $i' \in \known{A}$ such that for each dimension $d$ then $i'_d \neq \emptyset$ and has the form $i'_d = i_d + \Delta_d$, where $\Delta_d \in \{0, s_d\}$ for some fixed vector $s$ such that $s_d \in \Zbb_{>0}$.
    \label{fully-supported-tensor}
\end{definition}

Definition~\ref{fully-supported-tensor} says that every unknown entry forms a structure with $2^{D}-1$ known entries. Using this theorem, we obtain the following theorem regarding uniqueness of the recommendation/entry-completion result when there is sufficient data.

\begin{theorem} \textbf{Uniqueness of Tensor Completion Algorithm:}
    If $A$ is fully-supported, the result from $TCA(A, k)$ is uniquely determined even if there are distinct sets of latent vectors $\{\coeffsubtensorScaled_k\}$ that yield the same, unique, 
    $LLI(A, k)$.
\end{theorem}

\begin{proof}
For $A' = \TCA(A, k)$, since entry $i \in \known{A}$ has $A'[i] = A[i]$
by Theorem \ref{uniquess_A^{re}}, we need only prove uniqueness of any completion $i \in \unknown{A}$.
From the uniqueness result of Theorem \ref{uniquess_A^{re}}, $\TCA(A, k)$ admits two distinct scaling vectors $\coeffsubtensorScaled_k$ and $\coeffsubtensorScaled_k'$ that yield the same, unique, $LLI(A, k)$. From definition \ref{fully-supported-tensor}, there exists $2^D-1$ vector $i'$ such that $i'_d = i_d + \Delta_d$, where $\Delta_d \in \{0, s_d\}$ for some fixed vector $s$. 
From Theorem \ref{uniquess_A^{re}}, the scaling vector $\coeffsubtensorScaled_k'$ equals $\coeffsubtensorScaled_k \circ \coeffsubtensorAdd_k$ is equivalent to
\begin{equation}
        \prod\limits_{i'' \in \Ic:i' \in A[i'']}\hspace{-8pt} \coeffsubtensorAdd_{k, A[i'']}  ~= ~1 \qquad \forall i'.
\end{equation}
We now show that 
\begin{equation}
        \re{A}[i] = \prod_{i'' \in \Ic: i \in A[i'']} \coeffsubtensorScaled_{k,A[i'']}^{-1} = \prod_{i'' \in \Ic: i \in A[i'']} \coeffsubtensorScaled'^{-1}_{k,A[i'']}
\end{equation}
or equivalently from Theorem \ref{uniquess_A^{re}}
\begin{equation}
     \prod\limits_{i'' \in \Ic:i \in A[i'']}\hspace{-8pt} \coeffsubtensorAdd_{k, A[i'']}  ~= ~1 ~.
\end{equation}
Without loss of generality, we consider the case $D \equiv 0$ (mod $2)$, and the other case can be proven similarly. We define two sets $G_0$ and $G_1$ by the following. Except for $i' = i + s$, we divide $2^D-2$ remaining vectors $i'$ into two groups. For each $i' = i + \Delta$, if the number of $\Delta_d = 0$ equals $j$ modulo 2, then $i'$ goes to $G_j$ for $j \in \{0, 1\}$. We also denote $G_j \cap A[i'']= \{i' \in G_j \mid i' \in A[i'']\}$. Then
\begin{equation}
    \prod_{i' \in G_j}~\prod\limits_{i'' \in \Ic:i' \in A[i'']}\hspace{-8pt} \coeffsubtensorAdd_{k, A[i'']} ~~ = ~ 1 ~~\Leftrightarrow ~ \prod_{i'' \in \Ic}\hspace{-2pt} \coeffsubtensorAdd_{k, A[i'']}^{|G_j \cap A[i'']|} ~~=~ 1 ~.
\end{equation}

Consider $i'' \in \Ic$ that has $k$ non-null elements. Let's say that $i''$ has $C$ elements $i''_d = i_d$, then the remaining $k-C$ elements of $i''$ satisfies $i''_d = i_d + s_d$. Consider case 1 when $0 < C < k$. For each $i''$, we form a fixed $i'$ by replacing $m \leq D-k$ null elements of $i''$ such that $m+C$ is the number of elements in vector $i$ such that $i''_d = i_d$. For $j\in \{0,1\}$, if $m +C \equiv j (\text{mod } 2)$, then $i''$ forms $\binom{D-k}{m}$ numbers of $i'$ that belongs to $G_j$. So when we sum between $0 \leq m \leq D - k$, $|G_0 \cap A[i'']| = \sum_{m + C \equiv 0 (\text{mod } 2)} \binom{D-k}{m} = 2^{D-k-1}$ and $|G_1 \cap A[i'']| = \sum_{ m +C \equiv 1 (\text{mod } 2)} \binom{D-k}{m} = 2^{D-k-1}$. Thus,

\begin{equation} 
   \dfrac{\coeffsubtensorAdd_{k, A[i'']}^{|G_1 \cap A[i'']|}}{ \coeffsubtensorAdd_{k, A[i'']}^{|G_0 \cap A[i'']|}} ~=~ 1 ~.
\end{equation}
If $C = k$, we encounter the vector $i$ when forming $i'$. If $C = 0$, we encounter the vector $i + s$ when forming $i'$. Since we omit $i$ and $i+s$ from $G_0$, $|G_0 \cap A[i'']| = 2^{D-k-1} - 1$ and $|G_1 \cap A[i'']| = 2^{D-k-1}$ in either case of $C$. Thus,
\begin{equation} 
    \dfrac{\coeffsubtensorAdd_{k, A[i'']}^{|G_1 \cap A[i'']|}}{ \coeffsubtensorAdd_{k, A[i'']}^{|G_0 \cap A[i'']|}} ~=~ \coeffsubtensorAdd_{k, A[i'']} 
\end{equation}
and therefore
\begin{equation}
    \dfrac{\prod_{i'' \in \Ic} \coeffsubtensorAdd_{k, A[i'']}^{|G_1 \cap A[i'']|}} {\prod_{i'' \in \Ic} \coeffsubtensorAdd_{k, A[i'']}^{|G_0 \cap A[i'']|}} = 1 ~\Rightarrow \prod\limits_{i'' \in \Ic:i \in A[i'']}\hspace{-8pt} \coeffsubtensorAdd_{k, A[i'']}
    \prod\limits_{i'' \in \Ic:i' \in A[i'']}\hspace{-8pt} \coeffsubtensorAdd_{k, A[i'']}\,=  1 ~\Rightarrow \prod\limits_{i'' \in \Ic:i \in A[i'']}\hspace{-8pt} \coeffsubtensorAdd_{k, A[i'']} ~= 1 ~.
\end{equation}
This equality implies that $A'$ is unchanged, and thus uniquely determined.

\end{proof}

\subsubsection{Unit Consistency}
We now describe scale-invariance properties of the $\TCA(A,k)$ process.
\begin{theorem} \textbf{Unit Consistency:}
    Given a tensor $A$ and an arbitrary conformant positive latent vector $\coeffsubtensorScaled \in \Rbb^{|\Ic|}$, we have $\coeffsubtensorScaled *_k \TCA(A, k)\,=\,\TCA(\coeffsubtensorScaled *_k A, k)$.
    \label{unit-consistecy}
\end{theorem}
\begin{proof}
    Let $\coeffsubtensorScaled_k \leftarrow LLI(A, k)$. It can be shown that $\re{A} = LLI(\coeffsubtensorScaled_k *_k A, k) = LLI(A, k)$ for all $A$. We assume the unknown entries of $\re{A}$ are assigned the value of 1, i.e., $\re{A}[i] = 1$ for $i \in \unknown{A}$. The complete TCA process can then be defined as $\TCA(A, k) = \coeffsubtensorScaled_{k}^{(-1)} *_k \re{A}$,
    where $\coeffsubtensorScaled_{k}^{(-1)} = [\coeffsubtensorScaled_{k, A_i}^{-1}]_{i \in \Ic}$ is the inverse vector of $\coeffsubtensorScaled_k$. Now, using the uniqueness Theorem~\ref{uniquess_A^{re}}, we can subsume the latent vector $\coeffsubtensorScaled_k$ and deduce that
    \begin{equation}
        \coeffsubtensorScaled *_k \TCA(A, k) =  (\coeffsubtensorScaled_{k}^{(-1)} \circ \coeffsubtensorScaled) *_k \re{A} = \TCA(\coeffsubtensorScaled *_k A, k)
    \end{equation}
\end{proof} 
\subsubsection{Consensus Ordering}

When $k\,$=$\,D$$-$$1$, we establish how ordering from high to low rankings from known entries can be preserved in recommendation of unobserved entries. 
\begin{definition} \textbf{Ordering by $D$ dimension}:
Given a tensor $A$ and a permutation index vector $\gamma$ at dimension $D$ for each $\gamma_d \in [I_D]\cup\emptyset$, we define a set of vectors $\known{\Ic_{\gamma}}$ that preserves/follows ordering $\gamma$ in tensor $A$ if for $i \in \Ic_{\gamma}$ and $i \in \Zbb^{D-1}_{>0}$:
\begin{equation}
\begin{aligned}
    A[i, \gamma_d] \in \known{A}, \forall \gamma_d \neq \emptyset \text{ where } A[i, \gamma_{d_a}] < A[i, \gamma_{d_b}] \text{ when } a < b.
\end{aligned}
\label{consensus}
\end{equation}
Then the set of unobserved vectors $\unknown{\Ic_{\gamma}}$ satisfies that for $i\in\unknown{\Ic_{\gamma}}$,
\begin{equation}
\begin{aligned}
    A[i, \gamma_d] & \in \unknown{A} \quad \forall \gamma_d \neq \emptyset
\end{aligned}
\label{consensus2}
\end{equation}
\end{definition}

\begin{theorem} \textbf{Consensus Ordering:}
    Given a fully-supported tensor $A$, the obtained result $ \re{A} = \TCA(A, D-1)$, permutation index vector $\gamma$. Given $\known{\Ic_{\gamma}} \neq \emptyset$, then any unobserved vector $i \in \unknown{\Ic_{\gamma}}$ would satisfy $ \re{A}[i, \gamma_{d_a}] < \re{A}[i, \gamma_{d_b}] \text{ when } a < b.$
    \label{consensus ordering}
\end{theorem}
\begin{proof}
    Assuming the ordering for $i \in \known{\Ic_{\gamma}}$, $A[i, \gamma_{d_a}] < A[i,\gamma_{d_b}] \text{ when } a < b.$
    After the LLI process for $ \re{A} = LLI(A, D-1)$,
    \begin{align*}
        A[i, \gamma_{d_l}] =  A'[i, \gamma_{d_l}]  \cdot \coeffsubtensorScaled_{k,A[\emptyset, \gamma_{d_l}]}^{-1} \cdot  \prod_{d=1}^{D-1} \coeffsubtensorScaled_{k, A[\emptyset, i_d, \emptyset]}^{-1}
    \end{align*} 
    Substituting into the above inequality, we have:
    \begin{equation} 
        A'[i, \gamma_{d_a}] \cdot \coeffsubtensorScaled_{k,A[\emptyset, \gamma_{d_a}]}^{-1} \! 
         <\!  A'[i, \gamma_{d_b}]\cdot \coeffsubtensorScaled_{k,A[\emptyset, \gamma_{d_b}]}^{-1} 
       \label{ordering}
    \end{equation}
    Using the constraint from the problem setup:
    \begin{equation}
        \prod_{[i, \gamma_{d_l}] \in  \known{A[\emptyset, \gamma_{d_l}]}}\hspace{-2pt} \re{A}[i, \gamma_{d_l}] ~=~ 1~.
        \label{normalized_proof}
    \end{equation}
    Substituting (\ref{normalized_proof}) into (\ref{ordering}) gives $ \coeffsubtensorScaled_{k,A[\emptyset, \gamma_{d_a}]}^{-1} \! <\! \coeffsubtensorScaled_{k,A[\emptyset, \gamma_{d_b}]}^{-1}$. For any vector $i' \in \unknown{\Ic_{\gamma}} $ and $ \re{A} = \TCA(A, D-1)$, the following entry is uniquely determined,
    \begin{equation}
        \re{A}[i', \gamma_{d_l}] ~=~  \coeffsubtensorScaled_{k,A[\emptyset, \gamma_{d_l}]}^{-1} \cdot
        \prod_{d=1}^{D-1} \coeffsubtensorScaled_{k, A[\emptyset, i'_d, \emptyset]}^{-1},
    \end{equation}
    and we therefore deduce that $\re{A}[i, \gamma_{d_a}] < \re{A}[i, \gamma_{d_b}] \text{ when } a < b,$ thus completing the proof.
\end{proof}

Theorem~\ref{consensus ordering} applies even when we replace ordering at the $D$ dimension by ordering at any lower dimension $d\leq D$. We illustrate how this theorem provides the recommendation procedure from a given dataset with the observed user, product, and feature.

\section{Recommendation System Procedure}
\begin{definition}
    For a given user $u$ and product $p$ and vector $i = [u, p]$, denote $r_{u, p} =  \re{A}[i]$ as the recommendation result  from tensor A with $\re{A} = \TCA(A, D-1)$.
\end{definition}

We state the general recommendation framework from a $3$-dimensional tensor following Theorem~\ref{consensus ordering}. 
\begin{corollary} \textbf{Recommendation procedure for 3D Tensor:}
    Given a tensor $A \in \Rbb^{I_1\times I_2\times I_3}$ and $A' = \TCA(A, k=2)$. We retrieve $r_{u,p}$ by the following: Give permutation index vector of the features of users $\gamma_f \subset [I_2]$ and the non-empty set $\known{\Ic_{\gamma_f}}$, then any user-product vector $[u,p]\in \known{\Ic_{\gamma_f}}$ satisfies
        \begin{equation}
            \re{A}[u, \gamma_{f_a}, p] < \re{A}[u, {\gamma_{f_b}, p}] \text{ when } a < b.
        \end{equation}
        From here we could provide rating recommendation of user $u$ and product $p$ through maximization projection based on user's feature.
        \begin{equation}
            r_{u,p} = \max{(\re{A}[u,\emptyset,p])}
        \end{equation}
\label{RS3D}
\end{corollary}
Using Corollary~\ref{RS3D}, we obtain the method to retrieve the recommendation result of a user on a product by looking at the features that maximize their rating. This is remarkable since the proposed RS method comes from rigorous framework; \cite{CO1, CO2} have demonstrated the importance of using users' features for recommendation. Using Theorem~\ref{consensus ordering}, we derive the consensus ordering property in the context of a $2$-dimensional tensor and a $3$-dimensional tensor.

\begin{corollary} \textbf{Consensus Ordering for 2D tensor:}
    Given a 2D Tensor $A \in \Rbb^{m \times n}$ and $\MCA(A)$ with permutations index vectors of users $\gamma_u \subset [m]$ and of products $\gamma_p \subset [n]$. The following statements showcase the ranking consistency with respect to either user or product:
    \begin{enumerate}
        \item Given that $\known{\Ic_{\gamma_p}}$ is non-empty. Then any unobserved user $u \in \unknown{\Ic_{\gamma_p}}$ satisfies
        \begin{equation}
            r_{u, \gamma_{p_a}} < r_{u, \gamma_{p_b}} \text{ when } a < b.
        \end{equation}
        From here we could provide prediction on what product user $u$ will prefer/not prefer.
        \item Given that $\known{\Ic_{\gamma_u}}$ is non-empty. Then any unobserved product $p \in \unknown{\Ic_{\gamma_u}}$ satisfies
        \begin{equation}
            r_{\gamma_{u_a}, p} < r_{\gamma_{u_b}, p} \text{ when } a < b.
        \end{equation}
        From here we could provide prediction on which user will prefer/not prefer a product $p$.
    \end{enumerate}
    \label{CO2D Tensor}
\end{corollary}

\begin{corollary} \textbf{Consensus Ordering for 3D Tensor:}
Given a tensor $A \in \Rbb^{I_1\times I_2\times I_3}$ and $A' = \TCA(A, k=2)$. The following statements showcase the equivalency with corollary~\ref{CO2D Tensor}.
    \begin{enumerate}
        \item Give permutation index vector of products $\gamma_p \subset [I_3]$ and non-empty set of observed user-feature vector $\known{\Ic_{\gamma_p}}$. Then each unobserved user vector $u  \in \unknown{\Ic_{\gamma_p}}$ satisfies:
        \begin{equation}
            r_{u, \gamma_{p_a}} < r_{u, \gamma_{p_b}} \text{ when } a < b.
        \end{equation}
        \item Give permutation index vector of users $\gamma_u \subset [I_1]$ and non-empty set of observed product-feature vector $\known{\Ic_{\gamma_u}}$. Then each unobserved user vector $p \in \unknown{\Ic_{\gamma_u}}$ satisfies:
        \begin{equation}
            r_{\gamma_{u_a}, p} < r_{\gamma_{u_b}, p} \text{ when } a < b.
        \end{equation}
    \end{enumerate}
\label{COtensor}
\end{corollary}

\section{Experiment}
\subsection{2D Tensor Completion}
We experiment on three datasets: MovieLens10M, MovieLens1M, and Jester-2. Each of the two datasets MovieLens10M and MovieLens1M contains user ratings of films, along with features/attributes of users, and Jester 2 is a dataset containing user ratings of jokes. This permits both 2D and 3D tensor formulations of the latent invariance tensor completion problem, with user features representing the third tensor component. For this dataset, the observed ratings are in the interval $[1,5]$. The Jester-2 dataset provides ratings of jokes from a scale of $[-10,10]$, and we convert the data by a translational shift from the smallest rating among all users to obtain positive entries.

For each dataset, the ratio of training data to testing data is 4:1. The only hyperparameter is the stopping rate $\epsilon = 10^{-10}$, which only affects convergence accuracy. We incorporate the training data into a tensor $A$ and retrieve the recommendation through latent vectors. Specifically, we run throughout $MCA$ algorithm to obtain two latent vectors $\coeffsubtensorScaled_1$ and $\coeffsubtensorScaled_2$, from which a query seeking a recommendation of a specified product for a specified user can be obtained in $O(1)$ time.  Purely for consistency with conventions of previous methods, we use the root-mean-square error (RMSE) and mean absolute error (MAE) to evaluate the accuracy of our model. We compare our results with the well-established benchmarks: SVD, SVD with implicit ratings (SVD++), collaborative filtering (Co-Clustering, Slope One, and Normal Predictor), k-nearest neighbors (KNN) algorithms - KNNBasic, KNNwithMean, KNNZscore, and KNNBaseline - and non-negative matrix factorization (NMF). For KNNs, we set the number of neighbors as $k=25$ with the minimum threshold as $k=5$, with nearest neighbors determined using cosine similarity. We use a 5-fold cross-validation method to retrieve the mean and standard deviation. More details on our experimental set-up and codebase can be found at \href{https://github.com/tungnguyen1234/LLI}{https://github.com/tungnguyen1234/LLI}. 

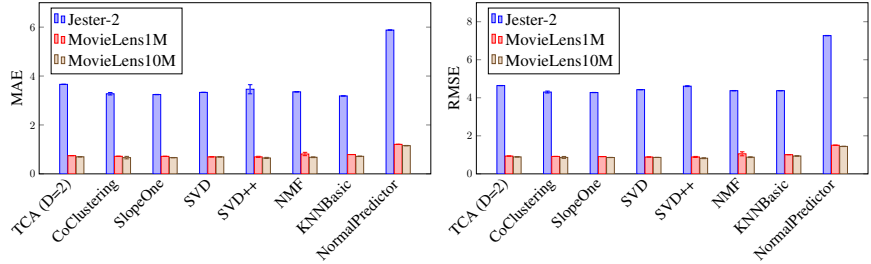
\begin{figure} \centering
    \resizebox{0.41\textwidth}{!}{
    \begin{tikzpicture}
    \begin{axis}[
        legend pos=north west,
        x tick label style={rotate=45, anchor=east, font=\Large},
        label style={font=\Large},
        width  = \textwidth,
        height = 7cm,
        major x tick style = transparent,
        ybar=1*\pgflinewidth,
        bar width=7pt,
        ymajorgrids = false,
        ylabel = {MAE},
        symbolic x coords={TCA (D=2), CoClustering, SlopeOne, SVD, SVD++, NMF, KNNBasic, NormalPredictor},
        xtick = data,
        scaled y ticks = false,
        axis line style={-},
        ymin=0,ymax=7,
        legend cell align=left,
        legend style={nodes={scale=1.5, transform shape}},
    ]
    \addplot+[error bars/.cd, y dir=both, y explicit]
         coordinates {
         (TCA (D=2), 3.6609) +- (0, 0.0045)
         (CoClustering, 3.2761) +- (0, 0.0497)
         (SlopeOne, 3.24346) +- (0, 0.00583)
         (SVD, 3.33185) +- (0, 0.00369)
         (SVD++, 3.4600) +- (0, 0.1868)
         (NMF, 3.35565) +- (0, 0.01075)
         (KNNBasic, 3.18679) +- (0, 0.00319)
         (NormalPredictor, 5.88031) +- (0, 0.00771)};

    \addplot+[error bars/.cd, y dir=both, y explicit]
       coordinates {
       (TCA (D=2), 0.745) +- (0, 0.0011)
       (CoClustering, 0.718) +- (0, 0.0016)
       (SlopeOne, 0.71621) +- (0, 0.00046)
       (SVD, 0.69619) +- (0, 0.00092)
       (SVD++, 0.6954) +- (0, 0.02298)
       (NMF, 0.8122) +- (0, 0.06803)
       (KNNBasic, 0.78978) +- (0, 0.00066)
       (NormalPredictor, 1.20675) +- (0, 8e-05)};
       
        \addplot+[error bars/.cd, y dir=both, y explicit]
     coordinates {(TCA (D=2), 0.69822) +- (0, 0.0052)
        (CoClustering, 0.664) +- (0, 0.051)
         (SlopeOne,0.6636) +- (0, 0.00595)
         (SVD, 0.694) +- (0, 0.003)
         (SVD++, 0.651) +- (0, 0.02298)
         (NMF, 0.6786) +- (0, 0.015)
         (KNNBasic, 0.717) +- (0, 0.00925)
         (NormalPredictor, 1.1543) +- (0, 0.00809)};
    
    \legend{Jester-2, MovieLens1M, MovieLens10M}
    \end{axis}
    \end{tikzpicture}
    }
    \resizebox{0.41\textwidth}{!}{
    \begin{tikzpicture}
        \begin{axis}[
        legend pos=north west,
        x tick label style={rotate=45, anchor=east, font=\Large},
        label style={font=\Large},
        width  = \textwidth,
        height = 7cm,
        major x tick style = transparent,
        ybar=1*\pgflinewidth,
        bar width=7pt,
        ymajorgrids = false,
        ylabel = {RMSE},
        symbolic x coords={TCA (D=2), CoClustering, SlopeOne, SVD, SVD++, NMF, 
        KNNBasic, 
        NormalPredictor},
        xtick = data,
        scaled y ticks = true,
        axis line style={-},
        ymin=0,ymax=9,
        legend cell align=left,
        legend style={nodes={scale=1.5, transform shape}},
    ]
        \addplot+[error bars/.cd, y dir=both, y explicit]
         coordinates {(TCA (D=2), 4.6422) +- (0, 0.0052)
            (CoClustering, 4.302) +- (0, 0.051)
             (SlopeOne, 4.27626) +- (0, 0.00595)
             (SVD, 4.42743) +- (0, 0.003)
             (SVD++, 4.6105) +- (0, 0.02298)
             (NMF, 4.37414) +- (0, 0.00925)
             (KNNBasic, 4.37414) +- (0, 0.00925)
             (NormalPredictor, 7.26971) +- (0, 0.00809)};
        \addplot+[error bars/.cd, y dir=both, y explicit]
         coordinates {
           (TCA (D=2), 0.9362)  +- (0, 0.001)
           (CoClustering, 0.916) +- (0, 0.00212)
           (SlopeOne, 0.90823) +- (0, 0.00049)
           (SVD, 0.8857) +- (0, 0.00119)
           (SVD++, 0.8895) +- (0, 0.02715)
           (NMF, 1.05451) +- (0, 0.1057)
           (KNNBasic, 1.00927) +- (0, 0.0007)
           (NormalPredictor, 1.50569) +- (0, 0.00022)};
\addplot+[error bars/.cd, y dir=both, y explicit]
         coordinates {(TCA (D=2), 0.893) +- (0, 0.0052)
            (CoClustering, 0.864) +- (0, 0.051)
             (SlopeOne, 0.8603) +- (0, 0.00595)
             (SVD, 0.864) +- (0, 0.003)
             (SVD++,0.822) +- (0, 0.02298)
             (NMF, 0.876) +- (0, 0.02)
             (KNNBasic, 0.9404) +- (0, 0.00925)
             (NormalPredictor, 1.443371) +- (0, 0.00809)};
        
        \legend{Jester-2, MovieLens1M, MovieLens10M}
    \end{axis}
\end{tikzpicture}
    }
    \caption{TCA when $D =2$, SVD, and other unitarily-invariant optimization methods implicitly minimize RMSE. The TCA approach yields results comparable to the state-of-the-art on the standard MovieLens-10M, MovieLens-1M, and Jester-2 benchmark dataset according to RMSE and MAE. }
    \label{RMSE and MAE dim 2}
\end{figure}

\subsection{3D Tensor Completion}
In this section we test $TCA(k=2)$ for $D=3$ on MovieLens1M dataset. The tensor has coordinates user, product, and user feature, and each respective triplet of user, product, and user's features corresponds to a rating score in the tensor. As there are 3 categorical features: age, occupation, and gender; we evaluate all 6 possible combinations of including either one, two, or all features in the feature dimension. In this dataset, the gender category has $2$ binary indices representing \textit{male} and \textit{female} and $20$ indices from $0$ to $20$ representing the occupation category (or \textit{occup} for short). As the dimension of age is from $0$ to $56$, we did stratification to group the user into 6 age groups, where group $i$ has the age range from $i \times 10$ to $(i+1) \times 10$.  The procedure can be described through the following example: assuming $TCA(A, 2)$ with tensor $A$ having feature dimension as age and gender, user 1 rates a score $4$ on product $2$, and user 1 is $male$ with index $1$ and is in \textit{age group 5}, then $A[u_1, p_2, \textit{male}] = 4$ and $A[u_1, p_2, \textit{age group 5}] = 4$. 

We construct a full 3-dimensional tensor with observed entries and then divided the entries into both a training tensor and a testing mask with a ratio of 4:1. As in the 2D case, the only parameter is the stopping rate as $\epsilon = 10^{-10}$. We run the $TCA(A, 2)$ to get the three latent vectors and construct a prediction tensor through projecting on a testing mask. Using Corollary~\ref{RS3D}, we perform a maximizing projection on the feature dimension based on the prediction tensor to obtain the predicted 2-dimensional RS for MovieLens1M.

\section{Results and Discussion}

The result for MovieLens1M, MovieLens10M ($D=2, D=3$), and Jester-2 ($D=2$) datasets are in Figure~\ref{RMSE and MAE dim 2} and Figure~\ref{RMSE and MAE}, respectively. In Figure~\ref{RMSE and MAE dim 2}, TCA gives comparable RMSE values of \textbf{0.893} for MovieLens10M, \textbf{0.936} for MovieLens1M, and \textbf{4.642} for Jester-2. This is significant because the competing methods are implicitly or explicitly designed to minimize squared error. Our approach, by contrast, is not designed to minimize squared error, and yet it performs nearly identically to methods that are tailored to minimize this measure. Our approach, by contrast, is not designed to minimize squared error, and yet it performs nearly identically to methods that are tailored to minimize this measure. In Figure~\ref{RMSE and MAE}, all MAE/RMSE values are comprable in each cases. In MovieLens1M, the TCA has comparable metric results with the MCA and and with other benchmarks. Additionally, Figure~\ref{mean_convg_rate} shows the fast convergence rate of the TCA: only 40 iterations over more than 1 million data points. If the conclusions of our analyses are correct, we should expect our latent invariance approach to be similarly competitive with methods that are tailored to minimize other measures of error -- {\em even according to those measures for which they are tailored to minimize}. 
 
In Figure~\ref{RMSE and MAE}, even though adding additional features is expected to yield more accurate predictions, one possibility is that the inherently arbitrary nature of MAE and RMSE as overall performance metrics contributes to the numerical variation. One possibility is to improve the testing retrieval process in the TCA to obtain better MAE/RMSE values. Nevertheless, our results provide evidence that unit consistency relates to something fundamental about the RS problem. Specifically, it achieves comparable performance with well-known benchmark methods without the need for any problem-specific hyperparameters.

\begin{figure}[t] \centering
    \begin{minipage}{0.45\textwidth}
        \begin{tikzpicture}
            \begin{axis}[ 
            width=\textwidth,
            line width=0.5,
            tick label style={font=\normalsize},
            legend style={nodes={scale=0.7, transform shape}},
            label style={font=\normalsize},
            grid style={white},
            height = 5cm,
            major x tick style = transparent,
            xlabel={Number of iterations},
            ylabel={Mean convergence rate},
            legend cell align=left,
      ]
        \addplot+[blue, very thick,  mark = o, error bars/.cd, y dir=both,y explicit] 
        coordinates{
            (5, 0.00024331205349881202) +- (0, 1.1818985512945801e-05)
            (6, 1.423022877133917e-05) +- (0, 6.906965950292943e-07)
            (7, 8.34212528388889e-07)  +- (0, 4.026579603078062e-08)
            (8, 4.893850658049814e-08)  +- (0, 2.3466890652912298e-09) 
            (9, 2.8707034527286623e-09) +- (0, 1.3650675034782012e-10)
            (10, 1.6893886289892635e-10) +- (0, 8.57475462162638e-12)
            (11, 1.0017786079841162e-11) +- (0, 4.625210262530766e-13)
        };
    
        \addplot+[green, very thick, mark = o, error bars/.cd, y dir=both,y explicit] 
        coordinates{
            (5, 0.003619122551754117) +- (0, 0.0002807920682244003)
            (6, 0.0010741017758846283) +- (0, 0.0002442183322273195)
            (7, 0.00038859128835611045) +- (0, 0.00020905723795294762)
            (8, 0.0001772734394762665) +- (0, 0.00015815903316251934)
            (9, 0.00010018399188993499) +- (0, 0.00011619772703852504)
            (10, 6.543874769704416e-05) +- (0, 8.565444295527413e-05)
            (11, 4.634271317627281e-05) +- (0, 6.393214425770566e-05)
            (15, 1.54856334120268e-05) +- ((0, 6.393214425770566e-05)
            (20, 4.663524578063516e-06) +- (0, 6.70756662657368e-06)
            (25, 1.4589452348445775e-06) +- (0, 2.098558070429135e-06)
            (30, 4.5993076014383405e-07) +- (0, 6.615681513721938e-07)
            (35, 1.452209232866153e-07) +- (0, 2.0888695928533707e-07)
        };

         \addplot+[red, very thick,mark = o,  error bars/.cd, y dir=both,y explicit] 
         coordinates{
            (5, 0.0038550831377506256) +- (0, 0.00011318801261950284)
            (6, 0.0010184949496760964) +- (0, 3.607536928029731e-05)
            (7, 0.00029271384119056165) +- (0, 1.2017482731607743e-05)
            (8, 8.844926924211904e-05) +-  (0, 4.124018687434727e-06)
            (9, 2.7741709345718846e-05) +- (0, 1.4408767583518056e-06)
            (10, 8.963429536379408e-06) +- (0, 5.088106149742089e-07)
            (11, 2.9654797799594235e-06) +- (0, 1.8079155950090353e-07)
            (12, 9.99568669612927e-07) +- (0, 6.448735945241424e-08)
            (13, 3.4180561669927556e-07) +- (0, 2.3075083177559463e-08)
            (14, 1.1816837286460213e-07) +- (0, 8.27387669488644e-09)
            (15, 4.1202429912345906e-08) +- (0, 2.9763065345633777e-09)
            (16, 1.445121711185493e-08) +- (0, 1.068303001616755e-09)
            (17, 5.09295805528609e-09) +-  (0, 3.8805941793285115e-10)
         };
         \addplot+[purple, very thick, mark = o,  error bars/.cd, y dir=both,y explicit] 
         coordinates{
            (5, 0.001366131822578609) +- (0, 6.221856165211648e-05)
            (6, 0.0003442417364567518) +- (0, 2.063162719423417e-05)
            (7, 9.784934809431434e-05) +- (0, 6.926702553755604e-06)
            (8, 2.9499402444344014e-05) +- (0, 2.380326122874976e-06)
            (9, 9.248193236999214e-06) +- (0, 8.318511959259922e-07)
            (10, 2.9878663099225378e-06) +- (0, 2.937594842933322e-07)
            (11, 9.88496594800381e-07) +- (0, 1.043799215949548e-07)
            (12, 3.3318954706373916e-07) +- (0, 3.723179275993971e-08)
            (13, 1.1393520793490097e-07) +- (0, 1.332240540108387e-08)
            (14, 3.938945880577194e-08)+- (0, 4.776925077720762e-09)
            (15, 1.3734143600174775e-08) +- (0, 1.7183713341140106e-09)
            (20, 7.569454091305161e-11) +- (0, 1.7183713341140106e-09)
            (25, 4.863150593337195e-07) +- (0, 1.0472298375696631e-11)
            (30, 1.5331025338127802e-07) +- (0, 1.2116030347897322e-06)
            (35, 4.840697442887176e-08) +- (0, 3.8195653928596585e-07)
         };

        \legend{1 feature, 2 features, 3 features, Overall mean}
        \end{axis}
        \end{tikzpicture}
        \caption{Convergence rate for the $D=3$ latent learning method, starting at step 5 for each category of features added.}
        \label{mean_convg_rate}
    \end{minipage}
    \hspace{2mm}
    \begin{minipage}{0.45\textwidth}
        \begin{tikzpicture}
        \begin{axis}[
            legend pos=north west,
            width=\textwidth,
            height=5cm,
            major x tick style=transparent,
            ybar=1*\pgflinewidth,
            bar width=7pt,
            ymajorgrids=false,
            xtick=data,
            ymin=0,
            ymax=1.6,
            ylabel={Performance metrics},
            xlabel={Feature combination},
            ]
            \addplot+[error bars/.cd, y dir=both, y explicit]
                coordinates {
                    (1, 0.736860) +- (0, 0.000428)
                    (2, 0.73686) +- (0, 0.00075)
                    (3, 0.73693) +- (0, 0.00075)
                    (4, 0.736135) +- (0, 0.00079)
                    (5, 0.736101) +- (0, 0.000525)
                    (6, 0.736201) +- (0,  0.00025)
                    (7, 0.735900) +- (0, 0.000726)
                };
            \addplot+[error bars/.cd, y dir=both, y explicit]
                coordinates {
                    (1, 0.92472) +- (0, 0.00018)
                    (2, 0.92483) +- (0, 0.00028) 
                    (3, 0.924896) +- (0,0.00028)
                    (4, 0.92387) +- (0, 0.00028)
                    (5, 0.92386) +- (0,  0.00017)
                    (6, 0.923934) +- (0, 0.00012)
                    (7, 0.923578) +- (0, 0.00038)
                };
            \legend{RMSE, MAE}
        \end{axis}
        \end{tikzpicture}
        \caption{Performance TCA when $D=3$. 1: occup, 2: age, 3: gender, 4: age-occup, 5: gender-occup, 6: age-gender, 7: all features.}
        \label{RMSE and MAE}
    \end{minipage}
\end{figure}

\section{Conclusion and Future Work}
In summary, we provide a tensor decomposition framework that utilizes a convex optimization model based on unit consistency to retrieve robust latent vectors for determining RS recommendations without the need for any hyperparameters. Our empirical results indicate high accuracy and low errors on standard datasets and surpass the performance of state-of-the-art benchmark methods. A topic of further research would be to extend our model into the probabilistic landscape of RS. Future work will examine implicit feedback involving information about user behaviors associated with the giving scores, e.g., measures of the number of clicks or duration of focus while watching a video.  

\bibliography{references}
\end{document}